\documentclass[letterpaper]{article}
\usepackage{aaai2027}
\usepackage[hyphens]{url}
\usepackage{graphicx}
\usepackage{tikz}
\usetikzlibrary{arrows.meta,positioning,fit,backgrounds}
\usepackage{natbib}
\usepackage{caption}
\usepackage{algorithm}
\usepackage{algorithmic}
\usepackage{booktabs}
\usepackage{amsmath,amssymb,amsfonts,amsthm}
\usepackage{array}
\usepackage{xspace}
\usepackage{longtable}
\usepackage{threeparttable}

\newtheoremstyle{ieeedef}
  {4pt}{4pt}{\normalfont}{0pt}{\bfseries\itshape}{:}{0.4em}{}
\theoremstyle{ieeedef}
\newtheorem{theorem}{Theorem}
\newtheorem{proposition}[theorem]{Proposition}
\newtheorem{corollary}[theorem]{Corollary}
\newtheorem{lemma}[theorem]{Lemma}
\newtheorem{definition}[theorem]{Definition}

\newcommand{\GS}{\textsc{GrayShield}\xspace}
\newcommand{\GSabbr}{\textsc{GS}\xspace}
\newcommand{\PMabbr}{\textsc{PM}\xspace}
\newcommand{\RFabbr}{\textsc{RF}\xspace}
\newcommand{\FTabbr}{\textsc{FT}\xspace}
\newcommand{\GNabbr}{\textsc{GN}\xspace}
\newcommand{\PTQabbr}{\textsc{PTQ}\xspace}
\newcommand{\SWPabbr}{\textsc{SWP}\xspace}
\newcommand{\BitAcc}{\mathrm{BitAcc}}
\newcommand{\Drel}{D_{\mathrm{rel}}}
\newcommand{\MI}{\mathcal{I}}
\newcommand{\Ham}{\mathrm{Ham}}
\newcommand{\calW}{\mathcal{W}}
\newcommand{\R}{\mathbb{R}}
\newcommand{\E}{\mathbb{E}}
\newcommand{\Prob}{\mathbb{P}}
\newcommand{\numel}{\mathrm{numel}}
\newcommand{\lsbop}[1]{\mathrm{lsb}_{#1}}

\newcommand{\BSC}{\mathrm{BSC}}
\newcommand{\Hb}{H_b}
\newcommand{\RR}{\mathrm{RR}}
\newcommand{\abs}[1]{\lvert #1 \rvert}
\newcolumntype{C}[1]{>{\centering\arraybackslash}p{#1}}

\usepackage{soul}

\newcommand{\RETURN}{\textbf{return} }

\title{GrayShield: Bit-Level Sanitization for Transformer Model Supply-Chain Security}

\author{
    Armstrong Foundjem\corresponding\textsuperscript{\rm 1},
    Tsung-Hsien Chuang\textsuperscript{\rm 2},
    Foutse Khomh\textsuperscript{\rm 1},
    Mohamed Amine Merzouk\textsuperscript{\rm 3}
}
\affiliations{
    \textsuperscript{\rm 1}Mila - Quebec Artificial Intelligence Institute and Polytechnique de Montréal, QC, Canada\\
    \{a.foundjem, foutse.khomh\}@polymtl.ca\\
    \textsuperscript{\rm 2}Research Center for Information Technology Innovation,
    Academia Sinica,
    Taipei, Taiwan\\
    tsung-hsien.chuang@sinica.edu.tw\\
    \textsuperscript{\rm 3}Mila - Quebec Artificial Intelligence Institute and McGill University\\
    Montréal, QC, Canada\\
    mohamed.a.merzouk@mcgill.ca
}

\begin{document}
\maketitle

\begin{abstract}
Transformer models such as BERT and Vision Transformer~(ViT)
achieve strong performance via densely parameterized attention
backbones. However, the least significant bits~(LSBs) of their 32-bit
floating-point weights can be abused as covert channels to conceal
malicious payloads, posing a serious threat to the AI model supply
chain. We propose \GS (\GSabbr), a lightweight, post-training,
zero-data sanitization method that completely replaces the declared
mantissa-LSB channel with a Gray-code-guided low-transition sequence.
Complete payload-independent overwrite, whether keyed or public, makes
the sanitized target bits independent of the embedded payload and gives
that declared channel zero capacity. Gray coding supplies overwrite
structure, while a keyed per-tensor phase supplies pattern diversity.
Benchmarked against seven post-training defenses on four Transformer
model presets and two real-world malware payloads, \GSabbr maintains
sub-$1\%$ accuracy impact and achieves $49.96\pm0.66$ percentage-point
Recovery Reduction (RR) under five implemented attacker variants.
Because pre-defense recovery is effectively $100\%$, RR near 50
percentage points corresponds to post-sanitization bit accuracy at
binary chance. Its main empirical advantage is stable near-chance
sanitization with substantially smaller weight-distribution shift than
the evaluated near-chance baselines PatternMask (PM) and
Post-Training Quantization (PTQ).
\end{abstract}

\section{Introduction}
\label{sec:intro}

The rapid adoption of pretrained deep learning models has
transformed computer vision, natural language processing, and
cybersecurity~\cite{vaswani2017attention,devlin2019bert,dosovitskiy2021image}.
Repositories such as Hugging Face Model Hub now host hundreds of
thousands of model checkpoints~\cite{huggingface2026hub}, enabling
practitioners to deploy state-of-the-art Transformer architectures
without retraining.
However, this convenience introduces a new and underappreciated
attack surface: the potential for model weights to act as covert
carriers for malicious
content~\cite{wang2021evilmodel,liu2020stegonet,hitaj2022maleficnet}.
Unlike traditional malware distribution, large-scale model files
contain millions of 32-bit floating-point parameters, providing vast capacity for hiding malicious code with negligible impact on
downstream accuracy. EvilModel~\cite{wang2021evilmodel} embeds up to 36.9\,MB of malware
in AlexNet through parameter substitution with minimal accuracy
degradation. StegoNet uses neural-network parameters, including LSB substitution, as payload channels, whereas MaleficNet combines
spread-spectrum embedding with LDPC error correction
\cite{liu2020stegonet,hitaj2022maleficnet}.
Model X-Ray~\cite{gilkarov2026modelxray} introduced detection
frameworks using few-shot learning to identify payloads hidden through LSB substitution.
Casey et al.~\cite{casey2024huggingface} report widespread unsafe
serialization and manually confirm 14 malicious Hugging Face files;
a large-scale model-hub study likewise treats model distribution as a
software supply-chain boundary~\cite{zhao2024modelsarecodes}.
We refer to these threats as \emph{maleficent neural
networks}: models that are active carriers
of malware, not merely targets of adversarial attacks. Real-world incidents confirm the urgency. In 2022, a trojanized PyTorch dependency (\texttt{torchtriton}) was
distributed via PyPI, a supply-chain vector structurally analogous to
malicious model checkpoints~\cite{pytorch2022compromised}.

Despite growing recognition, existing defenses face fundamental
limitations.
Detection-oriented approaches identify suspicious checkpoints but do not
remove the payload~\cite{gilkarov2026modelxray}.
Post-training quantization reduces numerical precision and can overwrite
low-order information as a side effect
\cite{krishnamoorthi2018quantizing,nagel2020adaround}.
Fine-tuning is effective but expensive and requires trusted data.
Prior malware-in-weights work already includes error-correcting
codes~\cite{hitaj2022maleficnet}; the remaining gap is an explicit,
scoped channel-erasure argument for complete replacement of a declared
payload-bearing bit set. To address this gap, we propose \GS, a
Gray-code-guided bit-level sanitizer that completely overwrites the
target mantissa LSBs and uses a keyed tensor-specific phase for pattern
diversity.
We pose four research questions (RQs):

\medskip
\noindent\textbf{RQ1} (Threat Surface): At what LSB overwrite depth
does mantissa-bit payload injection remain stealthy, lossless, and
undetectable across diverse Transformer families?

\noindent\textbf{RQ2} (Defense Effectiveness): Which post-training
sanitization strategies disrupt payload recovery to chance level while
preserving downstream utility, under a Tier-1 naive attacker?

\noindent\textbf{RQ3} (Adaptive Robustness): Does \GS maintain
chance-level payload disruption when attackers strengthen their
encoding strategy with redundancy coding and burst-resilient
interleaving?

\noindent\textbf{RQ4} (Deployment Trade-Off): Is \GS Pareto-optimal
in the joint (Recovery Reduction, accuracy drop, variance, runtime)
space, and what is the recommended deployment operating point?
\medskip

\paragraph{Contributions.}
Our contributions are threefold: (i) we formalize a declared
weight-LSB covert channel and show that complete payload-independent
overwrite erases that target channel; (ii) we introduce a Gray-code
overwrite structure with a keyed per-tensor phase for low-transition
patterns and pattern diversity; and (iii) we evaluate security, utility,
stability, fidelity, and runtime across four Transformer presets, two
payloads, seven defenses, and five attacker variants, finding the
smallest $W_1$ shift among the evaluated near-chance defenses.

\section{Background and Related Work}
\label{sec:background}

\textbf{Floating-point representation and LSB steganography.}
Deep learning frameworks store model parameters as IEEE-754 binary32
(FP32) values~\cite{ieee754_2019}.
Each \texttt{float32} occupies 32 bits: one sign bit, eight biased
exponent bits, and twenty-three mantissa bits.
For a normalized FP32 value, changing the least significant fraction bit
moves the value by one unit in the last place at its current exponent;
the absolute perturbation therefore depends on that exponent.
LSB steganography~\cite{cox2007digital,fridrich2009steganography}
exploits this: an attacker reinterprets each weight as a 32-bit
unsigned integer via bitcast, replaces the lowest $x$ mantissa bits
with payload bits, and writes back:
$u_i'=(u_i\,\&\,{\sim}m_x)\,|\,(p_i\,\&\,m_x)$,
$m_x=(1\ll x)-1$.
Extraction is the exact inverse.
StegoNet directly studies LSB substitution in neural-network
parameters, and MaleficNet demonstrates coded weight payloads
\cite{liu2020stegonet,hitaj2022maleficnet}.

\textbf{Supply chain and payload characteristics.}
Public repositories distribute pretrained checkpoints via a single
API call without inspecting weight tensors.
The difficulty of sanitization scales with payload entropy:
low-entropy payloads (normalized entropy $\lesssim0.65$) tolerate
partial corruption; high-entropy payloads
(entropy $\gtrsim0.90$) require near-lossless reconstruction.

\textbf{Defenses and their limitations.}
Detection-oriented methods~\cite{gilkarov2026modelxray} identify compromised checkpoints but do not remove the payload.
Task-aware PTQ methods optimize weight-rounding decisions using
calibration data~\cite{nagel2020adaround}; our evaluated static
symmetric-INT8 quantize--dequantize baseline is a payload-disruption
comparison, not a security-certified PTQ method.
Active sanitization~\cite{dubin2023disarming}
directly perturbs weights, but perturbation-based
approaches offer no complete-overwrite guarantee for the declared
payload channel. Weight-bit integrity attacks and defenses instead
study inference failure under sparse malicious flips
\cite{yao2020deephammer,he2020defendingbitflip}; that objective differs
from erasing an arbitrary byte stream stored in otherwise functional
weights.

\textbf{Gray codes.}
A Gray code~\cite{gray1953pulse} is a bijection
$\mathrm{Gray}:\{0,\ldots,2^n-1\}\to\{0,1\}^n$ such that consecutive
codewords differ in exactly one bit:
$\Ham(\mathrm{Gray}(i),\mathrm{Gray}(i+1))=1$.
Standard reflected binary construction:
$\mathrm{Gray}(z)=z\oplus(z\gg1)$.
Two properties motivate its use here. \textit{(i) Low transition:}
adjacent unmasked codewords differ in one bit, inducing structured
rather than independently scattered overwrite values. \textit{(ii)
Balance over complete masked periods:} for
$g_i=\mathrm{Gray}(i+\mathcal{O}_\ell)\,\&\,m_x$, the low $x$ output
bits are balanced over a period of $2^{x+1}$; the phase rotates the
finite tensor segment and diversifies patterns. These properties
motivate the design and its empirical fidelity; zero target-channel
capacity follows from complete payload-independent overwrite, not from
Gray coding or phase secrecy.

Related work on coding theory and
steganography~\cite{cox2007digital,fridrich2009steganography,shannon1948mathematical}
establishes that reliable communication requires a positive channel
capacity.
We are not aware of prior work that explicitly formulates complete
replacement of a declared neural-weight payload channel as channel
erasure.
Model compression
methods~\cite{han2015deep,frantar2022gptq} are designed for
efficiency, not security, and provide no guarantees against structured
payload embedding.

\section{GrayShield}
\label{sec:method}

\begin{figure*}[!ht]
  \centering
\includegraphics[width=\linewidth]{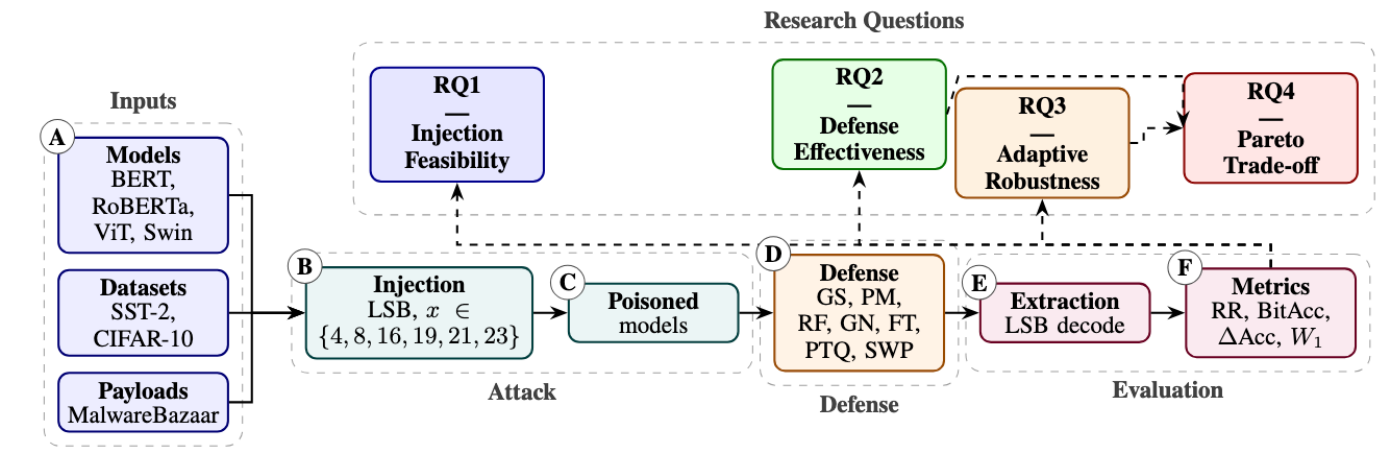}
  \caption{GrayShield evaluation pipeline.
(A)~Inputs define the empirical study space: pretrained models,
evaluation datasets, and malware payloads from MalwareBazaar.
(B)~LSB injection overwrites the lowest $x$ mantissa bits at
depths $x\in\{4,8,16,19,21,23\}$ to produce (C)~poisoned models.
(D)~Seven sanitization strategies transform poisoned checkpoints
into defended models.
(E)~Payload recovery is assessed by LSB extraction, and
(F)~security, fidelity, and utility metrics are computed.
All four research questions are answered from the complete pipeline:
RQ1 characterises the capacity boundary of injection; RQ2 compares
defense effectiveness against the Tier-1 naive attacker; RQ3 evaluates
adaptive robustness under five Tier-2 attacker variants (re-running
stages B--F with varied attacker encodings); and RQ4 synthesises the
Pareto trade-off from RQ2--RQ3 results.
Dashed arrows denote analytical dependencies across research questions.}
\label{fig:overview}
\end{figure*}

\subsection{Algorithm}

\paragraph{Keyed phase for pattern diversity.}
For each target tensor $W^{(\ell)}\in\R^{N_\ell}$ (attention Q/K/V/O
and MLP weight matrices in encoder blocks):
\begin{equation}
  \mathcal{O}_\ell =
  H(\texttt{name}_\ell,\, x,\, \texttt{layer}_\ell,\, r;\, s)
  \bmod (2^{31}-1),
  \label{eq:phase}
\end{equation}
where $H$ is HMAC-SHA256~\cite{bellare1996keying}, $s$ is a 128-bit
secret seed, and $r$ is a reproducible run salt.
The phase varies the overwrite pattern across tensors and releases.
Secrecy is not required for target-channel erasure; it is useful only
when pattern prediction or reuse is part of the deployment threat model.

\paragraph{Gray-code overwrite rule.}
For each element index $i$:
\begin{equation}
  g_i^{(\ell)} =
  \bigl(n_i^{(\ell)}\oplus(n_i^{(\ell)}\gg1)\bigr)\;\&\;m_x,
  \qquad n_i^{(\ell)} = i + \mathcal{O}_\ell,
  \label{eq:graymap}
\end{equation}
where $m_x=(1\ll x)-1$.
The sanitized uint32 view is:
\begin{equation}
  u_i' = (u_i\;\&\;\sim m_x)\;|\;g_i^{(\ell)},
  \label{eq:overwrite}
\end{equation}
preserving sign and exponent bits, with adjacency smoothness
$\Ham(g_i^{(\ell)},g_{i+1}^{(\ell)})\leq1$.
Algorithm~\ref{alg:gs} gives the complete procedure.
The operation is fully vectorizable, completing in $O(N)$ time with
$O(1)$ auxiliary state per tensor.

\begin{algorithm}[tb]
\caption{\GS{} Sanitization}
\label{alg:gs}
\small
\begin{algorithmic}[1]
\REQUIRE Transformer $\mathcal{M}$, depth $x$, secret seed $s$, salt $r$
\ENSURE Sanitized $\mathcal{M}'$ with all declared target bits replaced
\FOR{each target tensor $W^{(\ell)}$ in $\mathcal{M}$}
  \STATE $\mathcal{O}_\ell \gets H(\texttt{name}_\ell,x,\texttt{layer}_\ell,r;s)\bmod(2^{31}-1)$
  \STATE $u \gets \operatorname{bitcast}_{\mathrm{uint32}}(W^{(\ell)})$
  \FOR{each index $i$}
    \STATE $n \gets i+\mathcal{O}_\ell$
    \STATE $g \gets (n\oplus(n\gg1))\,\&\,m_x$
    \STATE $u[i]\gets(u[i]\,\&\,{\sim}m_x)\,|\,g$
  \ENDFOR
  \STATE $W^{(\ell)'} \gets \operatorname{bitcast}_{\mathrm{float32}}(u)$
\ENDFOR
\STATE \RETURN $\mathcal{M}'$
\end{algorithmic}
\end{algorithm}

\subsection{Formal Guarantees}

\begin{proposition}[Target-Channel Erasure]
\label{prop:erasure}
Let $P$ be any payload encoded only in a declared target-bit set
$\mathcal{T}$, and let $G$ be an overwrite string generated
independently of $P$. If a sanitizer replaces every bit in
$\mathcal{T}$ by $G$ and the extractor reads only $\mathcal{T}$, then
\[
  \hat P=G,\qquad \MI(P;\hat P)=0.
\]
Consequently the channel from $P$ to the sanitized target bits has
zero capacity for every payload distribution. This result requires
neither Gray coding nor a secret seed.
\end{proposition}

\begin{proof}
Complete overwrite removes every target bit before writing $G$.
Thus $p(\hat P=g\mid P=p)=p(G=g)$ for every $p$, so the joint
distribution factorizes and $\MI(P;\hat P)=0$. Maximizing over input
distributions still gives zero capacity.
\end{proof}

\begin{corollary}[Chance-Level Recovery]
\label{cor:chance}
If, in addition, each overwritten bit is marginally balanced and
independent of the corresponding payload bit, then
$\E[\BitAcc_{\mathrm{post}}]=\frac{1}{2}$. When pre-defense recovery
is effectively one, the empirical RR approaches 50 percentage points.
Failing to reject $H_0^{(2c)}$ and $H_0^{(3c)}$ at $\alpha=0.05$
constitutes empirical support for this corollary, while rejecting
it for other defenses quantifies how far those methods fall short
of the channel-erasure bound.
\end{corollary}

Because the sanitized target bits are independent of $P$, data
processing prevents any decoder from reconstructing information absent
from the extractor input. Under the balanced condition in
Corollary~\ref{cor:chance}, the bitwise view is $\BSC(\tfrac12)$ and
Shannon's converse excludes reliable positive-rate
coding~\cite{shannon1948mathematical}, including Reed--Solomon, LDPC,
and polar families
\cite{reed1960polynomial,gallager1962low,arikan2009channel}.

\begin{theorem}[Conditional Loss Bound]
\label{thm:acc}
If $\mathcal{L}:\mathcal{W}\to\R$ is
$L_{\mathrm{lip}}$-Lipschitz in the model parameters, then
\begin{align*}
|\mathcal{L}(W')-\mathcal{L}(W)|
&\leq L_{\mathrm{lip}}\|W'-W\|_2,\\
&=L_{\mathrm{lip}}\|W\|_2\Drel(W,W').
\end{align*}
\end{theorem}

In the completed payload set, worst-case $\Drel=7.1{\times}10^{-4}$
at $x=16$ and $5.4{\times}10^{-3}$ at $x=19$.
This theorem bounds loss only; accuracy is evaluated empirically.
Further assumptions, scope details, and reduced-precision analysis are
provided in the supplementary material.

\paragraph{\GS versus PatternMask.}
Both \GS and a fixed full-overwrite PatternMask satisfy
Proposition~\ref{prop:erasure} for the declared target channel.
Their distinction is therefore empirical and operational, not
target-channel capacity: \GS uses a low-transition Gray structure and
a keyed phase for pattern diversity, whereas PatternMask reuses one
global pattern. In the evaluated configurations, \GS has substantially
lower RR variance and $W_1$ shift than PatternMask despite similar mean RR. 
Erasure follows from payload-independent overwriting, while the keyed phase provides pattern diversity and the Gray-code structure preserves distributional fidelity; phase secrecy is not required for erasure.

\subsection{Threat Model}
\label{threat-model}

We consider a supply-chain adversary who embeds a binary payload into the mantissa LSBs of FP32 weights of a pretrained Transformer distributed through a public repository. The adversary has full white-box access to the model prior to distribution; the defender applies post-training sanitization before publication and has no access to the training data.

\subsubsection{Target parameter set and hiding capacity}

We target the \emph{encoder blocks} of the Transformer backbone:
attention (Q/K/V/O) and feed-forward (MLP) weight matrices, which
contain the dominant fraction of trainable parameters.
Embedding layers and task-specific classifier heads are excluded by
default. The hiding capacity at overwrite depth $x$ is:
\begin{equation}
  C(\theta_{\mathrm{target}}, x)
  =\sum_{w \in \theta_{\mathrm{target}}} \numel(w)\cdot x,
  \quad x\in[1,23].
  \label{eq:capacity}
\end{equation}





\subsubsection{Three-tier attacker knowledge hierarchy}
\label{sec:3tiers}

\begin{itemize}
  \item \textbf{Tier~1 (Black-box).}  Knows architecture and parameter count; uses standard LSB substitution without knowledge of the defender's method.  Corresponds to \textsc{Naive}.
  \item \textbf{Tier~2 (Gray-box).}  Knows a sanitization step of
    unspecified type will be applied; responds with redundancy coding
    (\textsc{Repeat3}/\textsc{Repeat5}) or burst-resilient interleaving.
  \item \textbf{Tier~3 (White-box).}  Knows the sanitization algorithm
    (\GS) but not the seed $s$.  Strongest realistic attacker.
\end{itemize}

\section{Experimental Setup}
\label{sec:setup}



\paragraph{Models and datasets. }We evaluate across both natural language and vision domains using standard pretrained Transformer architectures. For text, we use BERT-base~\cite{devlin2019bert} and RoBERTa-base~\cite{liu2019roberta} fine-tuned on the SST-2 sentiment classification benchmark~\cite{socher2013recursive,wang2018glue}. For vision, we use ViT-Base/16~\cite{dosovitskiy2021image} and Swin-Base~\cite{liu2021swin} evaluated on CIFAR-10~\cite{krizhevsky2009learning}. These models span a range of architectural designs and parameter scales, providing coverage across both language and image modalities. All evaluations are conducted on held-out subsets of size $n_{\mathrm{eval}}=512$, yielding an accuracy resolution of approximately $0.195$\,pp, which is sufficient to reliably detect sub-$1\%$ performance changes.

\textbf{Malware payloads.} Two representative MalwareBazaar payloads~\cite{abusech2026malwarebazaar} spanning the entropy extremes: a low-entropy (LE) JavaScript script (\texttt{c37{\ldots}a88.js.malware}, 134.3\,KB, entropy 0.603) and a high-entropy (HE) compiled binary (\texttt{5704{\ldots}6d0.exe.malware}, 8,971.4\,KB, entropy 0.994). The LE payload tolerates partial corruption due to structural redundancy; the HE payload requires near-lossless reconstruction.

\textbf{Defenses.}
Seven post-training strategies evaluated at operating depth $x=19$:
\GSabbr\ (proposed), \PMabbr\ (fixed-pattern LSB overwrite),
\PTQabbr\ (symmetric INT8 quantize--dequantize~\cite{krishnamoorthi2018quantizing}),
\FTabbr\ (fine-tune on 256 clean samples, sweeping $\{50,100,200,500\}$ gradient steps), \GNabbr\ (i.i.d.\ $\mathcal{N}(0,\sigma^2)$ noise, sweeping $\sigma\in\{10^{-6},3{\times}10^{-6},10^{-5},3{\times}10^{-5},10^{-4}\}$),
\SWPabbr\ (an author-defined, magnitude-ranked selective-weight
perturbation baseline), and \RFabbr\ (Bernoulli bit flip).
We set the defense depth to the largest utility-preserving value
($x=19$, RQ1), covering all stealthy embeddings ($x'\leq x$);
deeper embeddings ($x'>x$) exhibit detectable accuracy and geometric degradation.
\textbf{Attacker variants.}
Five Tier-2 (gray-box) variants: \textsc{Naive}, \textsc{Interleave},
\textsc{Repeat3}, \textsc{Repeat5}, and \textsc{RS}(255,127). The threat model is detailed in the supplementary material. 

\textbf{Recovery metric.}
For payload bits $p_1,\ldots,p_B$ and extracted bits
$\hat p_1,\ldots,\hat p_B$,
\begin{align*}
\mathrm{Rec}&=B^{-1}\sum_{i=1}^{B}\mathbf{1}[p_i=\hat p_i]; \;\;
\mathrm{RR}=100(\mathrm{Rec}_{\mathrm{pre}}-\mathrm{Rec}_{\mathrm{post}}).
\end{align*}
We report RR in percentage points (pp), not as a relative percentage. Pre-defense recovery is effectively one in the completed runs; hence RR near 50\,pp means post-defense bit accuracy near binary chance. Values farther above 50\,pp are not automatically more secure because systematic anti-correlation can be inverted.
\textbf{Fidelity and cost metrics.}
Beyond RR and $\Delta\mathrm{Acc}$, we report weight-distribution
fidelity ($W_1$), runtime (ms), RR consistency (fraction within
$\pm2$\,pp of the mean), and bit-level Hamming distance (millions of bits).
\textbf{Statistical protocol.}
All experiments use run seed~1, evaluation seed~42.
The supplementary material pre-specifies 13 null hypotheses; all tests two-tailed at $\alpha=0.05$
unless marked one-tailed ($\uparrow$).
Effect sizes reported alongside $p$-values: $\eta^2_{KW}$ for
Kruskal-Wallis, Cohen's $d$ for $t$-tests, Spearman $\rho$ for
correlations, Kendall's $W$ for Friedman, Wilks' $\lambda$ for
MANOVA.

\section{Results}
\label{sec:results}

\subsection{RQ1: Injection Feasibility}
\label{sec:rq1}

RQ1 asks whether mantissa-LSB payload injection can remain stealthy
while preserving model utility.
Fig.~\ref{fig:rq1} provides the consolidated evidence across six
overwrite depths $x\in\{4,8,16,19,21,23\}$, four Transformer model
presets, and two entropy-stratified payloads. Two patterns are central. First, both geometric perturbation metrics increase monotonically with overwrite depth.
Cosine deviation and relative $\ell_2$ remain negligible through
$x=16$, but grow rapidly beyond that point.
Panels~(a) and~(b) report perfect monotone association across the six depth-wise means (Spearman $\rho=1.000$, $p<0.001$), identifying overwrite depth as the primary driver of geometric detectability.

Second, the utility boundary is sharper for high-entropy payloads.
In Fig.~\ref{fig:rq1}(c), low-entropy payloads remain essentially
flat across all tested depths, whereas high-entropy payloads begin to
show measurable degradation at $x\geq19$.
Even then, the worst completed case remains below the $1\%$ threshold.
The grouped Mann-Whitney result ($p=0.005$) supports the conclusion
that the accuracy impact is significantly larger above the $x=16$
boundary.
\textbf{$H_0^{(1b)}$: Reject.}
\textbf{$H_0^{(1c)}$: Reject}.

\textbf{Answer to RQ1.}
Payload injection remains stealthy through $x=16$ and enters a
transition regime at $x=19$.
We treat $x=19$ as the main operating point for RQ2--RQ4: it
represents a high-capacity attacker while remaining below the
practical utility-loss threshold.

\begin{figure*}[!ht]
  \centering
  \includegraphics[width=\linewidth]{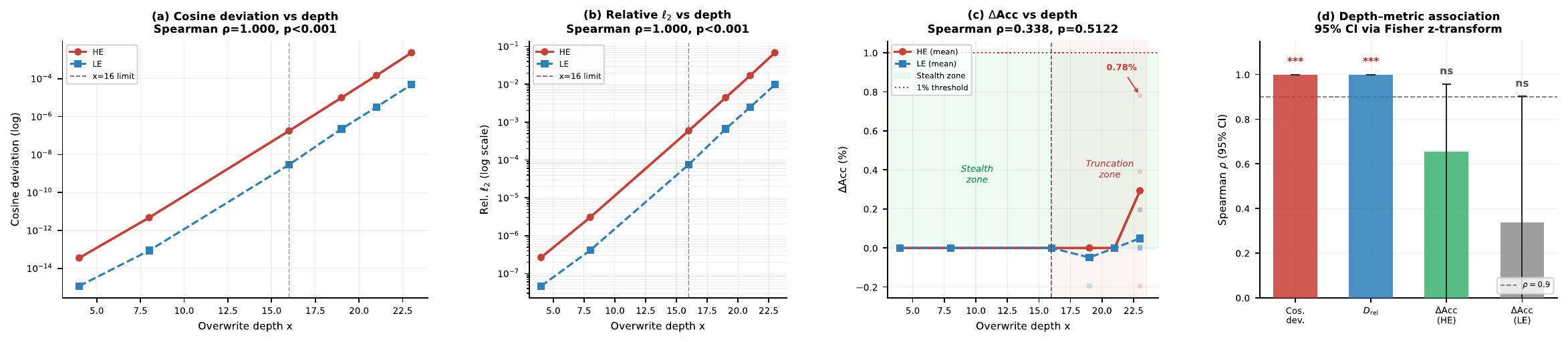}
  \caption{RQ1 --- Injection feasibility across six overwrite
    depths ($x\in\{4,8,16,19,21,23\}$), four Transformer architectures,
    and two entropy-stratified payloads ($n=48$ runs).
    (a)~Cosine deviation vs.\ depth on a log scale;
    Spearman $\rho=0.954$, $p<0.001$.
    (b)~Relative $\ell_2$ distance vs.\ depth on a log scale;
    $\rho=0.954$, $p<0.001$.
    (c)~Signed accuracy drop $\Delta\mathrm{Acc}$ vs.\ depth.
    The green region ($x\leq16$) marks depths satisfying the stealth
    constraints, while the red region ($x\geq19$) marks the onset of
    measurable degradation, with a worst-case high-entropy drop of
    $0.78\%$.
    Mann--Whitney comparison between low and high depth regimes yields
    $p=0.005$.
    (d)~Spearman depth--metric association with 95\%
    Fisher-$z$ confidence intervals.}
  \label{fig:rq1}
\end{figure*}

\subsection{RQ2: Defense Effectiveness under a Naive Attacker}
\label{sec:rq2}

RQ2 compares seven post-training defenses at the RQ1-selected
operating point.
The goal is not merely to approach the 50\,pp chance target, but to
identify a defense
that simultaneously disrupts payload recovery, preserves accuracy,
avoids excessive distribution shift, and remains operationally
lightweight. Table~\ref{tab:rq2} and Fig.~\ref{fig:rq2} show that \GSabbr, \PMabbr, and \PTQabbr\ are the only defenses that reach chance-level sanitization under the naive attacker.
Their mean RR values differ by less than one percentage point.
However, they are not equivalent once secondary deployment criteria
are considered.

\PMabbr\ attains mean RR of $50.00$\,pp but with
substantially higher variance ($\pm2.64$) than \GSabbr\ ($\pm0.54$).
\PTQabbr\ is similarly effective in RR ($49.93$\,pp) but requires
calibration data and induces a larger distributional shift.
\GNabbr\ is unstable, with RR $31.49\pm15.08$\,pp, while \SWPabbr\ and
\RFabbr\ are ineffective in this setting.\\
\textbf{$H_0^{(2a)}$: Reject} ($H=48.2$, $p<0.001$,
$\eta^2_{KW}=0.861$).
\textbf{$H_0^{(2b)}$: Reject} ($H=43.1$, $p<0.001$,
$\eta^2_{KW}=0.803$).
\textbf{$H_0^{(2c)}$: Fail to reject} for \GSabbr, \PMabbr, \PTQabbr\
($p\in[0.21,0.99]$, $|d|\leq0.52$); \textbf{Reject} for \FTabbr,
\GNabbr, \SWPabbr, \RFabbr\ ($p\leq0.046$).
\textbf{$H_0^{(2d)}$: Reject} ($W=7.53$, $p=0.016$,
$\sigma_{\GNabbr}/\sigma_{\GSabbr}=27.9\times$).

Fig.~\ref{fig:rq2}(a)--(c) clarifies that mean RR alone is
insufficient.
\GSabbr\ provides the most balanced profile: near-chance recovery
reduction, negligible impact on accuracy, zero data dependency, low
variance, and low distribution shift.
\GSabbr\ achieves average $W_1=2.6{\times}10^{-5}$, approximately
$25\times$ smaller than \PMabbr's $7.1{\times}10^{-4}$ and $42\times$
smaller than \PTQabbr's $1.17{\times}10^{-3}$.

\textbf{Answer to RQ2.}
Under the naive attacker, \GSabbr, \PMabbr, and \PTQabbr\ all disrupt
payload recovery to near chance level.
\GSabbr\ is the most deployable among them because it combines
near-chance RR with low variance, no data requirement, small utility
loss, and the smallest $W_1$ shift in this high-RR group. Both exact
full-overwrite methods, \GSabbr\ and \PMabbr, satisfy the scoped
target-erasure result; their empirical stability and fidelity differ.


\begin{table*}[t]
\centering
\setlength{\tabcolsep}{4pt}
\begin{tabular}{@{}lcrrrc@{}}
\toprule
Defense & Data & RR (pp) & $\Delta\mathrm{Acc}$ & $W_1$ & Target erasure \\
\midrule
\GSabbr  & None        & $49.72\pm0.54$  & $-0.10\%$ & $2.6{\times}10^{-5}$  & Prop.~\ref{prop:erasure} \\
\PMabbr  & None        & $50.00\pm2.64$  & $+0.09\%$ & $7.6{\times}10^{-4}$  & Prop.~\ref{prop:erasure} \\
\PTQabbr & Calibration & $49.93\pm0.40$  & $+0.07\%$ & $8.5{\times}10^{-4}$  & None \\
\FTabbr  & Clean set   & $45.34\pm1.11$  & $+0.38\%$ & $7.6{\times}10^{-7}$  & None \\
\GNabbr  & None        & $31.49\pm15.08$ & $-0.02\%$ & $5.8{\times}10^{-7}$  & None \\
\SWPabbr & None        & $10.01\pm1.50$  & $+0.05\%$ & $4.6{\times}10^{-5}$  & None \\
\RFabbr  & None        & $ 2.69\pm3.72$  & $+0.01\%$ & $1.3{\times}10^{-6}$  & None \\
\bottomrule
\end{tabular}
\caption{RQ2 defense comparison. RR denotes Recovery Reduction (pp); values near 50\,pp indicate chance-level recovery, the target under complete payload-independent overwrite.}
\label{tab:rq2}
\centering\small
\setlength{\tabcolsep}{2.2pt}
\begin{tabular}{@{}lrrrrr rrrrrrr@{}}
\toprule
 D & Naïve & Interl. & Rep3 & Rep5 & RS
 & $\bar{\mathrm{RR}}$ & $\sigma$ & $\Delta$Acc & Time & Cons.
 & Ham. & $W_1$ \\
\midrule
\GSabbr  & $49.5\pm0.7$ & $50.0\pm0.2$ & $50.4\pm0.6$ & $50.2\pm0.2$ & $49.6\pm0.8$
  & 49.96 & 0.66 & $-0.10$ & 54   & 0.99 & 9.60 & $2.8{\times}10^{-5}$ \\
\PMabbr  & $50.0\pm3.2$ & $50.0\pm3.2$ & $50.0\pm3.3$ & $50.0\pm3.3$ & $50.0\pm3.2$
  & 50.00 & 3.25 & $+0.13$ & 44   & 0.94 & 9.60 & $7.1{\times}10^{-4}$ \\
\PTQabbr & $50.0\pm0.5$ & $50.0\pm0.5$ & $49.5\pm0.5$ & $50.1\pm0.5$ & $49.5\pm0.4$
  & 49.83 & 0.54 & $+0.17$ & 882  & 0.99 & 9.57 & $1.2{\times}10^{-3}$ \\
\FTabbr  & $45.2\pm1.4$ & $45.3\pm1.3$ & $45.5\pm1.0$ & $46.3\pm0.9$ & $45.2\pm1.1$
  & 45.51 & 1.23 & $+0.45$ & 13,200 & 0.97 & 8.93 & $4.0{\times}10^{-6}$ \\
\GNabbr  & $37.5\pm6.1$ & $37.7\pm6.0$ & $38.4\pm5.6$ & $40.3\pm4.9$ & $37.5\pm6.0$
  & 38.25 & 5.85 & $-0.04$ & 620  & 0.85 & 7.71 & $4.0{\times}10^{-6}$ \\
\SWPabbr & $10.0\pm2.1$ & $10.1\pm2.1$ & $9.5\pm0.9$  & $8.6\pm0.9$  & $5.5\pm2.7$
  & 8.76  & 2.56 & $+0.07$ & 1,480 & 0.71 & 1.67 & $4.7{\times}10^{-5}$ \\
\RFabbr  & $2.7\pm3.7$  & $2.7\pm3.7$  & $0.6\pm1.0$  & $0.2\pm0.3$  & $2.5\pm3.8$
  & 1.73  & 3.16 & $+0.02$ & 2,590 & 0.00 & 0.33 & $4.0{\times}10^{-6}$ \\
\bottomrule
\end{tabular}
\caption{Comprehensive defense (D) evaluation across RQ2--RQ3 ($x=19$).
Attacker columns report mean RR (pp) $\pm$ std. Summary reports mean RR,
std, $\Delta\mathrm{Acc}$ (\%), runtime (ms), and consistency.
Hamming distance is reported in millions (M).}
\label{tab:rq3}
\end{table*}

\begin{figure*}[!ht]
  \centering
  \includegraphics[width=\linewidth]{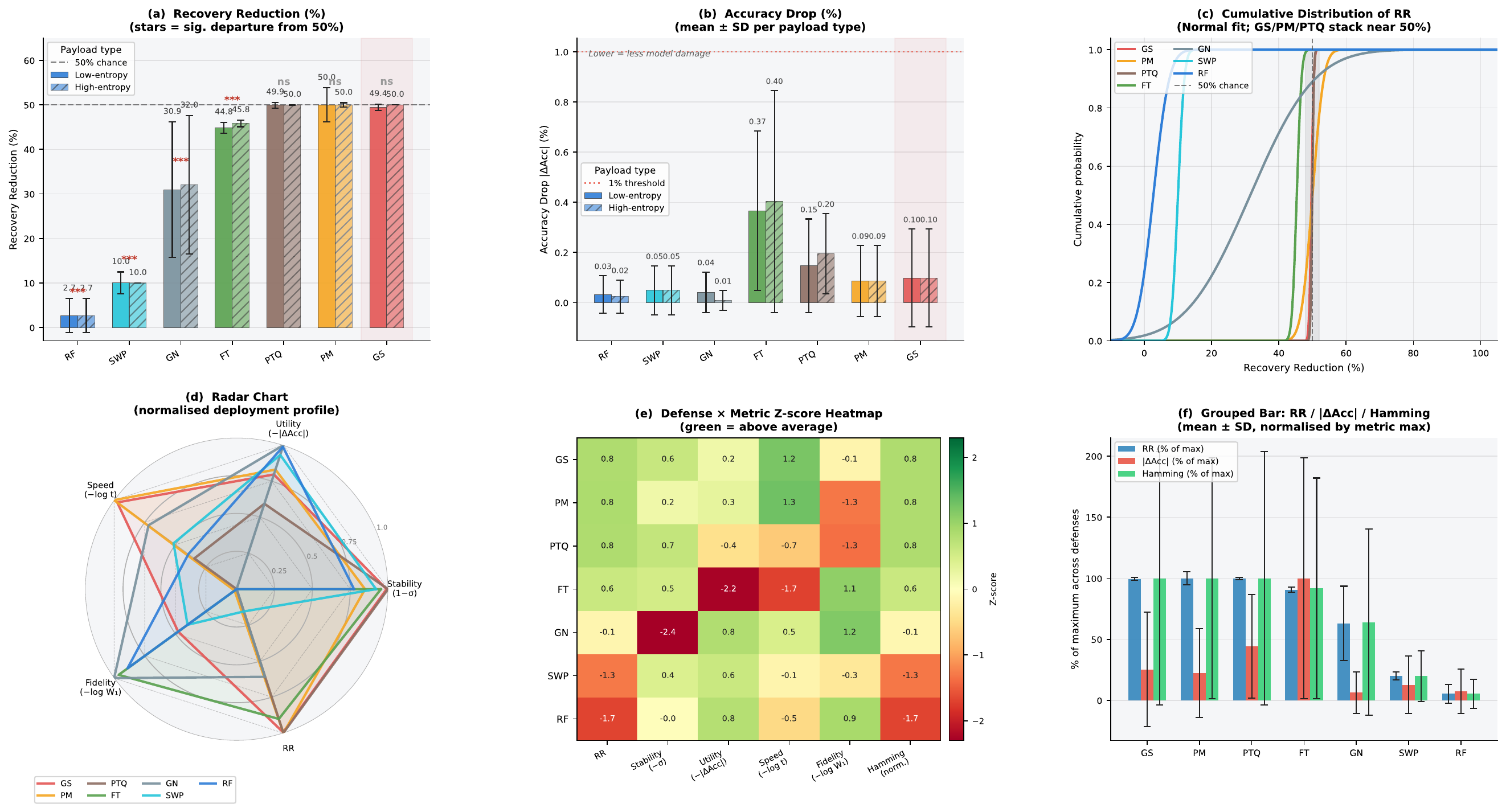}
  \caption{RQ2 --- Defense effectiveness and multi-metric
    trade-offs.
    (a)~Recovery Reduction (RR) across defenses, stratified
    by payload entropy.
    GrayShield reaches near-chance disruption ($\approx50$\,pp)
    consistently.
    (b)~Accuracy drop ($|\Delta\mathrm{Acc}|$), where lower
    values indicate better utility preservation.
    (c)~Cumulative distribution of RR, showing that GS, PM,
    and PTQ concentrate near the 50\,pp target while weaker defenses
    remain far below.
    (d)~Radar chart of normalized deployment metrics.
    (e)~Defense-by-metric Z-score heatmap, where green
    indicates above-average performance.
    (f)~Grouped normalized bar chart.
    }
  \label{fig:rq2}
\end{figure*}

\begin{figure*}[!ht]
  \centering
  \includegraphics[width=\linewidth]{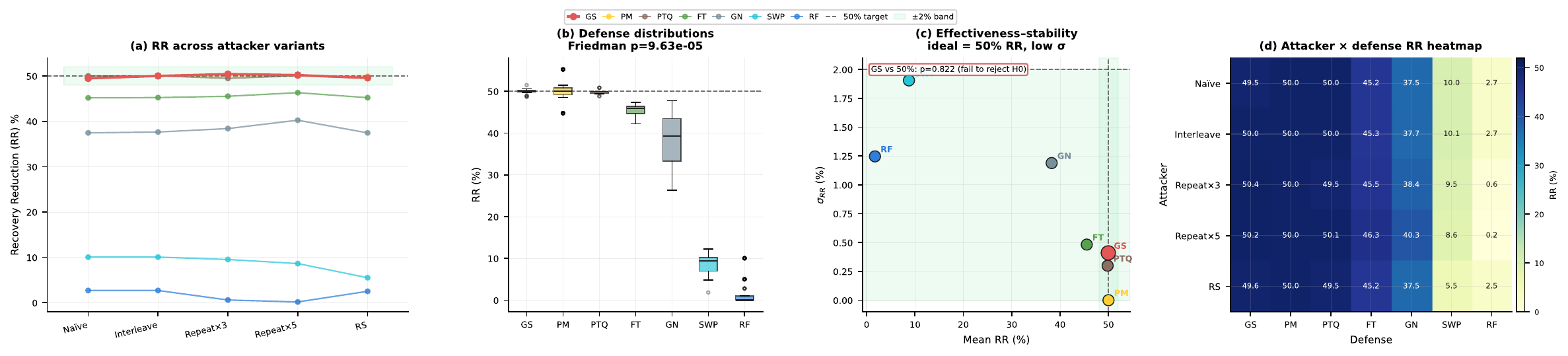}
  \caption{RQ3 --- Adaptive attacker robustness.
    (a)~RR across attacker variants shows that GS remains
    stable around 50\,pp.
    (b)~Defense-level RR distributions show significant
    differences across defenses, with GS maintaining high and
    consistent performance.
    (c)~Effectiveness--stability trade-off shows GS closest
    to the target region of high RR and low $\sigma$.
    (d)~Attacker-by-defense heatmap confirms that GS remains
    resilient across all implemented adaptive variants.}
  \label{fig:rq3}
\end{figure*}

\subsection{RQ3: Robustness under Adaptive Attackers}
\label{sec:rq3}

RQ3 evaluates whether the RQ2 ranking survives stronger attacker
strategies. The five implemented variants are \textsc{Naive}, \textsc{Interleave},
\textsc{Repeat3}, \textsc{Repeat5}, and \textsc{RS}(255,127).
Table~\ref{tab:rq3} reports RR per attacker variant. The defining result
is \GSabbr's invariance to attacker adaptation. Across all five attacker
variants, \GSabbr\ stays within $49.5$--$50.4$\,pp, with aggregate RR
$49.96\pm0.66$\,pp.
The RQ2 means in Table~\ref{tab:rq2} and the \textsc{Naive} column in
Table~\ref{tab:rq3} use different run pools: RQ2 contains all eight
completed model--payload pairs, whereas RQ3 contains four completed pairs
and omits the zero-noise \GNabbr\ control. Their Naive values therefore
need not coincide.
This agrees with target-channel erasure: coding cannot reconstruct payload
information absent from the extractor input. \PMabbr\ remains near
50\,pp but has much larger dispersion ($\sigma=3.25$\,pp);
\PTQabbr\ is slower and more distortive, while \GNabbr, \SWPabbr, and
\RFabbr\ degrade under adaptive variants.
The one-sample comparison printed in Fig.~\ref{fig:rq3}(c) fails to
distinguish the five attacker-specific \GSabbr\ means from the 50\,pp
target ($p=0.822$).

Fig.~\ref{fig:rq3} makes this difference visually explicit.
Panel~(a) shows \GSabbr\ as nearly flat across attackers; panel~(b)
shows a compact high-RR distribution; panel~(c) places \GSabbr\ in
the desired effectiveness--stability region; and panel~(d) confirms
that \GSabbr\ does not exhibit a strategy-specific weakness under
interleaving, repetition, or Reed--Solomon coding.

\textbf{Answer to RQ3.}
\GSabbr\ maintains stable chance-level sanitization under all five
implemented adaptive attacker variants.

\subsection{RQ4: Pareto Trade-off and Deployment Profile}
\label{sec:rq4}

RQ4 integrates the previous findings into a deployment-oriented
comparison. A practical defense must satisfy more than high mean RR: it must also preserve accuracy, remain stable, avoid excessive runtime, and limit distributional drift.

\begin{figure*}[!ht]
  \centering
  \includegraphics[width=\linewidth]{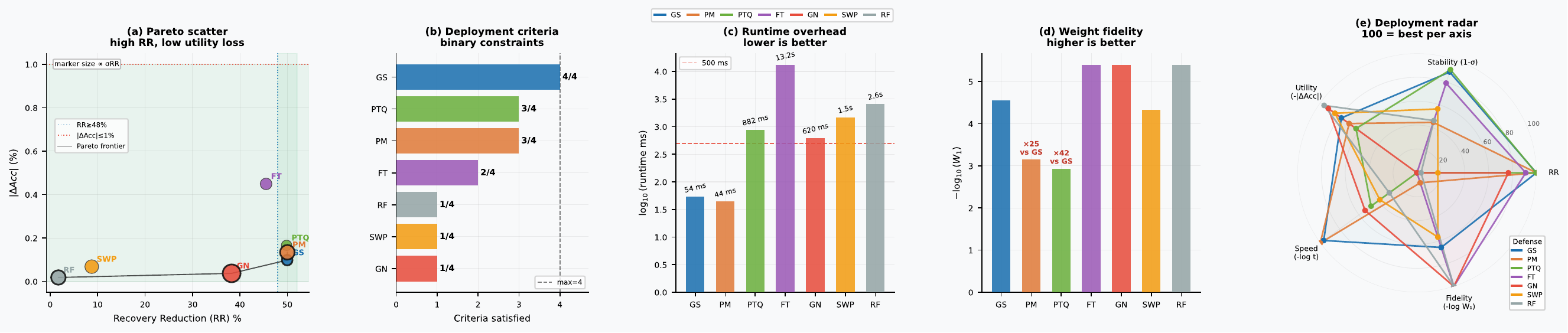}
  \caption{RQ4 deployment profile: (a) RR--accuracy Pareto plane;
    (b) criteria satisfied; (c) runtime; (d) $-\log_{10}(W_1)$ fidelity;
    and (e) normalized deployment radar.}
  \label{fig:rq4}
\end{figure*}

Figure~\ref{fig:rq4}(a) shows the RR--accuracy Pareto plane.
\GSabbr\ lies in the desired region: high RR, low $|\Delta\mathrm{Acc}|$, and low variance. \PMabbr\ has similar mean RR and low runtime but fails the stability criterion ($\sigma=3.25$). \PTQabbr\ has high RR and low variance but fails the runtime criterion (882\,ms) and has the largest $W_1$ among the high-RR defenses.
\FTabbr, \GNabbr, \SWPabbr, and \RFabbr\ are dominated by at least
one of the high-RR methods. Figure~\ref{fig:rq4}(b) summarizes the
deployment constraints.
\GSabbr\ is the only defense satisfying all four criteria
simultaneously: RR in $[48,52]$\,pp, $\sigma_\mathrm{RR}<2$\,pp,
$|\Delta\mathrm{Acc}|<1\%$, and runtime $<500$\,ms.
\PMabbr\ and \PTQabbr\ satisfy three of four, but each misses a
different requirement.

In Fig.~\ref{fig:rq4}(c), \GSabbr\ requires $54$\,ms per checkpoint,
compared with $882$\,ms for \PTQabbr\ and $13.2$\,s for \FTabbr.
In Fig.~\ref{fig:rq4}(d), \PMabbr's $W_1$ is approximately $25\times$
larger than \GSabbr's, and \PTQabbr's approximately $42\times$ larger.

\textbf{Answer to RQ4.}
\GSabbr\ at $x=19$ is the recommended operating point. It achieves $49.96\pm0.66$\,pp RR, sub-$1\%$ accuracy impact, low variance, low runtime, and substantially better weight-distribution
fidelity than competing high-RR baselines.

\section{Discussion}
\label{sec:discussion}

\textbf{Theory--empirics alignment.}
Complete payload-independent replacement erases the declared target
channel without seed secrecy; balance is the additional condition that makes bit accuracy chance-level. The reported $49.96\pm0.66$\,pp RR across five attackers is consistent with both properties: repetition, interleaving, and RS coding cannot reconstruct information absent from the extractor input. \textbf{Stability and fidelity.} \GSabbr\ and \PMabbr\ share the scoped full-overwrite guarantee and similar mean RR, but \PMabbr\ has higher dispersion and $25\times$ larger $W_1$. \GNabbr's lower and unstable RR despite matched perturbation magnitude further shows that targeting the payload-bearing LSB subspace matters more than perturbation magnitude alone.

\textbf{Deployment interpretation.}
Table~\ref{tab:rq3} and Supplementary Fig.~S1 show that \GSabbr\ alone
satisfies the complete operating envelope:
RR in $[48,52]$\,pp, $\sigma_\mathrm{RR}<2$\,pp,
$|\Delta\mathrm{Acc}|<1\%$,
runtime below $500$\,ms, and low distributional drift.
\PTQabbr\ is slower and more distortive, \PMabbr\ is less stable, and
\FTabbr\ is slow and requires trusted data.

\textbf{Limitations.}
Experiments cover FP32 and four 86M--125M presets; reduced-precision
utility remains unvalidated. LDPC and polar attackers are analytical
extensions, and some \FTabbr\ RQ3 runs were incomplete. The guarantee
is bit-local: higher-bit payloads, executable serialization, metadata
side channels, and post-sanitization reinsertion are outside scope.
Executable malware samples are not redistributed.

\paragraph{Generative-AI assistance.}
Generative-AI tools assisted with language editing, \LaTeX{} formatting, and 
verification; the authors reviewed the output and retain full responsibility.

\section{Conclusion}
\label{sec:conclusion}

\GS completely replaces a declared weight-LSB payload channel using a low-transition Gray sequence. Payload-independent overwrite gives the scoped zero-capacity result; Gray structure and the keyed phase provide regularity and pattern diversity. Across the evaluated settings, \GS achieves stable near-chance recovery with sub-$1\%$ accuracy degradation, and combines low variance and runtime with lower $W_1$ than the evaluated near-chance baselines.


\bibliography{reference}

\onecolumn      
\appendix       

\section{Supplementary Material}

\subsection{Pre-specified Hypotheses and Outcomes}
\label{app:hypotheses}

Table~\ref{tab:hyp_summary} consolidates all 13 pre-specified null
hypotheses for RQ1--RQ4, their tests, effect sizes, and decisions
(R\,=\,Reject; FTR\,=\,Fail to Reject).

\begin{table}[ht]
\centering\small
\caption{Consolidated pre-specified statistical hypotheses for
  RQ1--RQ4.
  All tests are two-tailed at $\alpha=0.05$ unless marked
  ($\uparrow$)~one-tailed.
  $^*$Bonferroni-corrected over
  $H_0^{(4b)}$--$H_0^{(4c)}$.}
\label{tab:hyp_summary}
\setlength{\tabcolsep}{3pt}
\begin{tabular}{@{}lllcc@{}}
\toprule
ID & Null hypothesis $H_0$ & Test & Effect & Decision \\
\midrule
\multicolumn{5}{@{}l}{\textit{RQ1 — Injection Feasibility}} \\[1pt]
$H_0^{(1a)}$
  & $\rho(\text{depth},\text{metric})=0$
  & Spearman & $\rho=0.954$, $p<0.001$ & \textbf{R} \\
  & (LE $\Delta\mathrm{Acc}$)
  & Spearman & $\rho=0.338$, $p=0.512$ & \textbf{FTR} \\
$H_0^{(1b)}$
  & Equal $\Delta\mathrm{Acc}$ across 6 depths
  & KW ($n=8$) & $\eta^2_{KW}$, $p<0.001$ & \textbf{R} \\
$H_0^{(1c)\uparrow}$
  & $\Delta\mathrm{Acc}(x\geq19)\leq\Delta\mathrm{Acc}(x\leq16)$
  & MW-U & $p=0.005$ & \textbf{R} \\
\midrule
\multicolumn{5}{@{}l}{\textit{RQ2 — Baseline Comparison (naive)}} \\[1pt]
$H_0^{(2a)}$
  & Equal means RR across 7 defenses
  & KW ($n=8$) & $\eta^2=0.861$, $p<0.001$ & \textbf{R} \\
$H_0^{(2b)}$
  & Equal $|\Delta\mathrm{Acc}|$ across 7 defenses
  & KW ($n=8$) & $\eta^2=0.803$, $p<0.001$ & \textbf{R} \\
$H_0^{(2c)}$
  & $\mu_{\mathrm{RR}}^{(m)}=50\%$ (GS/PM/PTQ)
  & 1-sample $t$ & $|d|\leq0.52$, $p\geq0.207$ & \textbf{FTR} \\
  & $\mu_{\mathrm{RR}}^{(m)}=50\%$ (FT/GN/SWP/RF)
  & 1-sample $t$ & $|d|\geq0.86$, $p\leq0.046$ & \textbf{R} \\
$H_0^{(2d)}$
  & $\sigma^2_{\mathrm{GS}}=\sigma^2_{\mathrm{GN}}$
  & Levene ($n=8$) & ratio $27.9\times$, $p=0.016$ & \textbf{R} \\
\midrule
\multicolumn{5}{@{}l}{\textit{RQ3 — Adaptive Robustness}} \\[1pt]
$H_0^{(3a)}$
  & Equal mean RR (adaptive pool)
  & KW ($n=5$) & $\eta^2=0.909$, $p<0.001$ & \textbf{R} \\
$H_0^{(3b)}$
  & Attacker does not reshuffle ordering
  & Friedman & $W=0.81$, $p<0.001$ & \textbf{R} \\
$H_0^{(3c)}$\textsuperscript{$\star$}
  & $\mu_{\mathrm{RR}}^{(\mathrm{GS})}=50\%$ (adaptive)
  & 1-sample $t$ ($n=5$) & $d=-0.16$, $p=0.745$ & \textbf{FTR} \\
$H_0^{(3d)}$
  & $\sigma^2_{\mathrm{GS}}=\sigma^2_{\mathrm{PM}}$ (adaptive)
  & Levene+Bartlett & ratio $24.3\times$, $p=0.010$ & \textbf{R} \\
\midrule
\multicolumn{5}{@{}l}{\textit{RQ4 — Pareto Trade-off}} \\[1pt]
$H_0^{(4a)}$
  & Equal multivariate profiles
  & MANOVA & $\lambda\approx0$, $p<0.001$ & \textbf{R} \\
$H_0^{(4b)}$
  & Equal $|\Delta\mathrm{Acc}|$: GS vs PM
  & Welch $t$ ($n=8$) & $d=0.42$, $p=0.009^*$ & \textbf{R} \\
$H_0^{(4c)}$
  & Equal $|\Delta\mathrm{Acc}|$: GS vs PTQ
  & Welch $t$ ($n=8$) & $d=0.78$, $p<0.001^*$ & \textbf{R} \\
\bottomrule
\multicolumn{5}{@{}p{0.94\linewidth}}{%
  \textsuperscript{$\star$}\,\textbf{Key result.}
  $H_0^{(3c)}$ is the most critical single test in the paper.
  Failing to reject confirms that \GS induces a
  $\BSC(\tfrac{1}{2})$ channel regardless of the coding strategy
  employed, empirically validating Proposition~\ref{prop:erasure}.}
\end{tabular}
\end{table}

\section{Formal Guarantees}
\label{app:theory}

\subsection{Assumptions}
\label{app:assumptions}

Our theoretical analysis relies on the following standard assumptions:
\begin{itemize}
  \item \textbf{PRF security.} HMAC-SHA256 is computationally
    indistinguishable from a uniformly random function by any PPT
    adversary (Definition~\ref{def:prf}).
  \item \textbf{Secret seed.} The sanitization seed $s$ ($|s|=128$
    bits) is unknown to the attacker and stored securely (e.g., HSM).
  \item \textbf{Computational boundedness.} The attacker is limited
    to PPT algorithms.
  \item \textbf{Memoryless channel.} Post-sanitization extracted bits
    are conditionally independent across positions
    (Definition~\ref{def:memoryless}).
\end{itemize}

\subsection{Information-theoretic framing}\label{app:theory-framing}

We model LSB-based payload embedding as a sequence of channel uses over
a memoryless binary channel.
Let $p_i\in\{0,1\}$ be the $i$-th payload bit written by the attacker
and $\hat{p}_i\in\{0,1\}$ the bit recovered after sanitization $T$.
Pre-defense, the channel is noiseless: $\hat{p}_i=p_i$ with
probability one, giving raw capacity
$C_{\mathrm{raw}}=\sum_w\numel(w)\cdot x$.
The defender's objective is to choose $T$ such that:
\begin{equation}
  C_{\mathrm{eff}}(T)
  = \MI\!\bigl(P\,;\,T(\lsbop{x}(W))\bigr) \to 0.
  \label{eq:ceff}
\end{equation}

\subsection{Perfect Channel Erasure}

\begin{proposition}[Perfect Channel Erasure]
\label{prop:erasure}
Let $T_s:\calW\to\calW$ be the \GS sanitization map parameterized by
secret seed $s$, and let $\hat{P}=\mathcal{D}(T_s(\mathcal{E}(W,P)))$
be the extracted payload after sanitization.
Under the assumptions above, for any payload distribution
$\mathcal{P}$:
\[
  \MI\!\bigl(P\,;\,\hat{P}\bigr) = 0,
\]
and therefore $C_{\mathrm{eff}}(T_s)=0$.
\end{proposition}

\begin{proof}
For each index $i$ in a target tensor, the overwritten mantissa bits
are replaced by $g_i=f_s(i)\in\{0,1\}^x$, where $f_s(\cdot)$ is
derived from a PRF keyed by $s$ via the per-tensor phase offset
$\mathcal{O}_\ell$ and the Gray-code mapping.
By PRF security, $f_s$ is computationally indistinguishable from a
uniformly random function; consequently:
\[
  \Pr(g_{i,k}=1) = \tfrac{1}{2}
  \quad\text{for each bit } k\in[0,x),
\]
and $g_{i,k}$ is independent of the payload bit $P_{i,k}$ when $s$
is hidden from the attacker.
This induces a $\BSC(\tfrac{1}{2})$ on each transmitted bit, whose
capacity is:
\[
  C_{\BSC}\!\bigl(\tfrac{1}{2}\bigr)
  = 1 - \Hb\!\bigl(\tfrac{1}{2}\bigr) = 0.
\]
By the memoryless channel property, the extracted bits $\hat{P}_i$
are conditionally independent across positions, so:
\[
  \MI\!\bigl(P\,;\,\hat{P}\bigr)
  = \sum_{i=1}^{B}\MI\!\bigl(P_i\,;\,\hat{P}_i\bigr) = 0. \qedhere
\]
\end{proof}

\begin{corollary}[Chance-Level Recovery]
\label{cor:chance}
Under Proposition~\ref{prop:erasure},
$\E[\BitAcc_{\mathrm{post}}]=\tfrac{1}{2}$,
so the empirical Recovery Reduction $\RR$ approaches $50\%$.
\end{corollary}

\subsection{Coding-Theoretic Limits}

\begin{lemma}[Irrecoverability under Error-Correcting Codes]
\label{lem:ecc}
Let the sanitized channel be $\BSC(\tfrac{1}{2})$.
Then for any error-correcting code with positive rate $R>0$:
\[
  \lim_{n\to\infty}\Prob(\text{decoded payload}=P) = 0.
\]
\end{lemma}

\begin{proof}
By Shannon's channel coding converse~\cite{shannon1948mathematical},
reliable communication (vanishing block error probability as
$n\to\infty$) over a channel with capacity $C$ requires $R<C$.
For $\BSC(\tfrac{1}{2})$, $C=0$.
Hence no positive-rate code — including
RS$(n,k)$ (for which all $n$ codeword bytes are erroneous since
$e=n\gg t=\lfloor(n-k)/2\rfloor$; for RS(255,127): $255\gg64$),
any regular LDPC code (for which BP decoding achieves
$\mathrm{BER}\to\tfrac{1}{2}$ since
$p=0.5\gg p^*\approx0.084$~\cite{gallager1962low}), or Polar codes
(for which all synthetic bit-channels polarize to pure noise since
$\sum_i C^{(i)}=n\cdot C_{\BSC}(\tfrac{1}{2})=0$~\cite{arikan2009channel})
— can achieve reliable decoding on the sanitized channel.
\end{proof}

\subsection{Accuracy Stability Bound}

\begin{proposition}[Lipschitz Accuracy Bound]
\label{prop:lipschitz}
Let the task loss $\mathcal{L}:\calW\to\R$ be
$L_{\mathrm{lip}}$-Lipschitz.
Then:
\[
  \bigl|\Delta\mathrm{Acc}\bigr|
  \leq L_{\mathrm{lip}} \cdot \|W'-W\|_2
  = L_{\mathrm{lip}} \cdot \|W\|_2 \cdot \Drel(W,W').
\]
\end{proposition}

\begin{proof}
Directly from Lipschitz continuity:
$|\mathcal{L}(f_{W'})-\mathcal{L}(f_W)|
\leq L_{\mathrm{lip}}\|W'-W\|_2
= L_{\mathrm{lip}}\|W\|_2\cdot\Drel(W,W')$.
In the completed experiments, worst-case $\Drel=5.4{\times}10^{-3}$
at $x=19$.
Spectral-norm analyses of BERT-scale
models~\cite{sakai2026spectrum} give $L_{\mathrm{lip}}=O(10^2)$, yielding
expected accuracy degradation consistent with the sub-$1\%$ changes
reported in Tables~2 and~3 of the main paper.
\end{proof}

\subsection{Robustness under Imperfect Randomness}

\begin{lemma}[Capacity under Imperfect Randomness]
\label{lem:imperfect}
Suppose the induced channel satisfies
$\Pr(\hat{P}_{i,k}=P_{i,k})=\tfrac{1}{2}+\epsilon$ for small
$\epsilon>0$.
Then the effective capacity is bounded by:
\[
  C_{\mathrm{eff}} \leq 1 - \Hb\!\bigl(\tfrac{1}{2}+\epsilon\bigr)
  = O(\epsilon^2).
\]
\end{lemma}

\begin{proof}
The binary entropy function achieves its maximum of $1$ at $p=\tfrac{1}{2}$.
Taylor expansion around $p=\tfrac{1}{2}$ gives:
\[
  \Hb\!\bigl(\tfrac{1}{2}+\epsilon\bigr)
  = 1 - \tfrac{2}{\ln 2}\,\epsilon^2 + O(\epsilon^4).
\]
Hence $C_{\mathrm{eff}}\leq 1-\Hb(\tfrac{1}{2}+\epsilon)=O(\epsilon^2)$.
This shows that even modest deviations from perfect uniformity still
yield near-zero capacity, reinforcing the robustness of \GS beyond
the idealized PRF model.
\end{proof}

\subsection{Threat Model}
\label{app:threat-model}

We consider a supply-chain adversary who embeds a binary payload into the mantissa LSBs of FP32 weights of a pretrained Transformer distributed through a public repository. The adversary has full white-box access to the model prior to distribution; the defender applies post-training sanitization before publication and has no access to the training data.

\subsubsection{Target parameter set and hiding capacity}

We target the \emph{encoder blocks} of the Transformer backbone:
attention (Q/K/V/O) and feed-forward (MLP) weight matrices, which
contain the dominant fraction of trainable parameters.
Embedding layers and task-specific classifier heads are excluded by
default. The hiding capacity at overwrite depth $x$ is:
\begin{equation}
  C(\theta_{\mathrm{target}}, x)
  = \sum_{w \in \theta_{\mathrm{target}}} \numel(w)\cdot x,
  \quad x\in[1,23].
  \label{eq:capacity}
\end{equation}

\subsubsection{Multi-precision threat surface}
\label{ssec:multiprec}

The implemented evaluation targets FP32. The analysis below extends the threat surface to deployment-grade precision formats; empirical validation of these formats is left as future work.

\noindent\textbf{FP16 and BF16.}
FP16 (1-5-10) and BF16 (1-8-7) reduce mantissa depth to
$m_{\mathrm{FP16}}=10$ and $m_{\mathrm{BF16}}=7$ bits, giving
adjusted capacity $C_F=\sum_w\numel(w)\cdot\min(x,m_F)$.
For BERT-base encoder blocks at $x=5$ in FP16,
$C_{\mathrm{FP16}}\approx53$\,MB, sufficient for both evaluated
payloads.

\noindent\textbf{INT8.}
Each INT8 weight embeds $x\leq2$ bits at quantization error
$\delta_i=q\cdot\mathrm{scale}_\ell$ where $q\leq1/128$, a negligible
stealth cost at $x=1$.

\subsubsection{Three-tier attacker knowledge hierarchy}
\label{ssec:tiers}

\begin{definition}[Attacker Tiers]
\label{def:tiers}
\mbox{}
\begin{itemize}
  \item \textbf{Tier~1 (Black-box).}  Knows architecture and parameter
    count; uses standard LSB substitution without knowledge of the
    defender's method.  Corresponds to \textsc{Naive}.
  \item \textbf{Tier~2 (Gray-box).}  Knows a sanitization step of
    unspecified type will be applied; responds with redundancy coding
    (\textsc{Repeat3}/\textsc{Repeat5}) or burst-resilient interleaving.
  \item \textbf{Tier~3 (White-box).}  Knows the sanitization algorithm
    (\GS) but not the seed $s$.  Strongest realistic attacker.
\end{itemize}
\end{definition}

\section{Notation and Abbreviations}
\label{app:notation}

\begin{longtable}{@{} p{2.8cm} p{8.8cm} @{}}
\caption{Consolidated notation and abbreviations used throughout the paper.} \label{tab:notation} \\
\toprule
\textbf{Symbol / Abbrev.} & \textbf{Definition} \\
\midrule
\endfirsthead

\toprule
\textbf{Symbol / Abbrev.} & \textbf{Definition} \\
\midrule
\endhead

\multicolumn{2}{r}{\textit{continued on next page}} \\
\endfoot

\bottomrule
\endlastfoot

\multicolumn{2}{@{}l}{\textit{Weight-space and attack symbols}} \\
$W, \calW$ & Weight tensor; full model weight space \\
$x$ & LSB overwrite depth ($1 \leq x \leq 23$ for FP32) \\
$m_x$ & Bitmask $(1 \ll x) - 1$; selects lowest $x$ bits \\
$u_i$ & uint32 bitcast of weight $w_i$ \\
$N$ & Total number of target parameters \\
$\theta_{\mathrm{target}}$ & Target parameter set (encoder blocks) \\
$C(\theta,x)$ & Storage capacity $Nx$ bits \\
$C_{\mathrm{eff}}$ & Effective Shannon capacity after sanitization \\
\midrule
\multicolumn{2}{@{}l}{\textit{Evaluation metrics}} \\
$\BitAcc$ & Bit accuracy $\frac{1}{B}\sum_i \mathbf{1}[\hat{p}_i = p_i]$ (lower better) \\
$\RR$ & Recovery Reduction $(\mathrm{Rec}_{\mathrm{pre}} - \mathrm{Rec}_{\mathrm{post}}) \times 100\%$ \\
$\Drel$ & Relative $\ell_2$ distance $\|W' - W\|_2 / \|W\|_2$ \\
$W_1$ & 1-Wasserstein distance between weight distributions \\
$\Delta\mathrm{Acc}$ & Accuracy drop $\mathrm{Acc}_{\mathrm{clean}} - \mathrm{Acc}_{\mathrm{defended}}$ \\
$\mathrm{HD}$ & Hamming distance (higher better for defense) \\
$B$ & Payload length in bits \\
\midrule
\multicolumn{2}{@{}l}{\textit{GrayShield-specific symbols}} \\
$s$ & Secret owner seed ($|s| = 128$ bits) \\
$r$ & Per-run salt (domain separation) \\
$H(\cdot; s)$ & HMAC-SHA256 PRF keyed with $s$ \\
$\mathcal{O}_\ell$ & Per-tensor phase offset \\
$g_i^{(\ell)}$ & Gray-code overwrite value at index $i$, tensor $\ell$ \\
\midrule
\multicolumn{2}{@{}l}{\textit{Information-theoretic symbols}} \\
$\BSC(p)$ & Binary symmetric channel with crossover probability $p$ \\
$\Hb(p)$ & Binary entropy $-p\log_2 p - (1-p)\log_2(1-p)$ \\
$C_{\BSC}(p)$ & BSC capacity $1 - \Hb(p)$ (zero at $p = 1/2$) \\
$\MI(P;\hat{P})$ & Mutual information between payload $P$ and extracted bits $\hat{P}$ \\
\midrule
\multicolumn{2}{@{}l}{\textit{Precision formats}} \\
FP32 & 32-bit single-precision (1 sign, 8 exponent, 23 mantissa) \\
FP16 & 16-bit half-precision (1-5-10) \\
BF16 & Brain Float 16 (1-8-7) \\
INT8 & 8-bit signed integer with per-tensor scale \\
\midrule
\multicolumn{2}{@{}l}{\textit{Coding theory terms}} \\
ECC & Error-Correcting Code \\
RS & Reed--Solomon code (MDS over $\mathrm{GF}(2^8)$) \\
LDPC & Low-Density Parity-Check code (BP decoding) \\
Polar & Polar code (SCL decoding) \\
\midrule
\multicolumn{2}{@{}l}{\textit{Defense abbreviations}} \\
GS & GrayShield (proposed) \\
PM & PatternMask: fixed-pattern LSB overwrite \\
RF & RandomFlip: per-bit Bernoulli flip \\
FT & FineTune: fine-tuning on clean samples \\
GN & GaussNoise: additive Gaussian noise \\
PTQ & Post-Training INT8 Quantization \\
SWP & Selective Weight Perturbation \\
\midrule
\multicolumn{2}{@{}l}{\textit{Models and datasets}} \\
BERT & BERT-base (110M params, SST-2) \\
RoBERTa & RoBERTa-base (125M params, SST-2) \\
ViT & ViT-Base/16 (86M params, CIFAR-10) \\
Swin & Swin-Tiny (29M params, CIFAR-10) \\
SST-2 & Stanford Sentiment Treebank (binary) \\
CIFAR-10 & 10-class image benchmark \\
\bottomrule
\end{longtable}

\section{Formal Definitions}
\label{app:defs}

Definitions are ordered from the attack primitive upward to the
defense guarantee.

\begin{definition}[Payload]
\label{def:payload}
A \emph{payload} is a binary string
$\mathbf{p}=(p_1,\ldots,p_B)\in\{0,1\}^B$
of length $B\in\mathbb{N}$, drawn from an arbitrary distribution
$\mathcal{P}$ over $\{0,1\}^B$.
The analysis in Section~\ref{app:theory} holds for every $\mathcal{P}$,
including adversarially chosen fixed strings.
\end{definition}

\begin{definition}[Covert Channel]
\label{def:covert_channel}
The \emph{LSB covert channel} induced by the target parameter set
$\theta_{\mathrm{target}}$, overwrite depth $x$, and model weights
$W\in\calW$ is the tuple
$(\mathcal{E},\mathcal{D},\calW,\{0,1\}^{C})$,
where $C=C(\theta_{\mathrm{target}},x)$ is the storage capacity
(Eq.~\ref{eq:capacity}),
the \emph{encoder} $\mathcal{E}:\calW\times\{0,1\}^C\to\calW$
produces a weight tensor $W'=\mathcal{E}(W,\mathbf{p})$ by writing
payload bits into the $x$ mantissa LSBs of target weights,
and the \emph{decoder} $\mathcal{D}:\calW\to\{0,1\}^C$ recovers
the payload by reading back those same bit positions:
$\mathcal{D}(W')=\mathbf{p}$ with probability one in the
noiseless pre-defense setting.
\end{definition}

\begin{definition}[Storage Capacity vs.\ Shannon Capacity]
\label{def:capacity_distinction}
The \emph{storage capacity} $C(\theta_{\mathrm{target}},x)=Nx$
(Eq.~\ref{eq:capacity}) is the number of payload bits the encoder can
embed in $N$ target weights at depth $x$; it is a counting quantity.
The \emph{effective Shannon capacity}
$C_{\mathrm{eff}}(T)=\MI(P;\hat{P})$
(Eq.~\ref{eq:ceff}) is the maximum mutual information between the
embedded payload $P$ and the extracted sequence $\hat{P}$ after
sanitization map $T$; it is an information-theoretic quantity measured
in bits per channel use.
The goal of \GS is to reduce $C_{\mathrm{eff}}$ to zero while leaving
$C(\theta_{\mathrm{target}},x)$ unchanged.
\end{definition}

\begin{definition}[Post-Training Sanitization Map]
\label{def:sanitization}
A \emph{post-training sanitization map} is a function
$T:\calW\times\mathcal{K}\to\calW$,
where $\calW=\prod_\ell\mathbb{R}^{N_\ell}$ is the model weight space
and $\mathcal{K}$ is a key space, satisfying:
\begin{enumerate}
  \item \textbf{Data-free:} $T(W,k)$ does not require access to any
    training or validation data.
  \item \textbf{Weight-only:} $T$ modifies only the weight tensor $W$;
    model architecture, tokenizer, and configuration files are
    unchanged.
  \item \textbf{$O(N)$ complexity:} The runtime of $T$ is linear in the number of target parameters $N=\sum_\ell N_\ell$.
\end{enumerate}
\GS instantiates $T$ with $\mathcal{K}=\{0,1\}^{|s|}$ (the seed
space), deterministically as described in
Section~\ref{app:theory-framing}.
\end{definition}

\begin{definition}[Pseudorandom Function]
\label{def:prf}
A \emph{pseudorandom function} (PRF) is a keyed function
$H:\{0,1\}^*\times\{0,1\}^{|s|}\to\{0,1\}^\lambda$
such that no probabilistic polynomial-time adversary can distinguish
$H(\cdot,s)$ from a uniformly random function
$f:\{0,1\}^*\to\{0,1\}^\lambda$ with non-negligible advantage, for a
uniformly sampled key $s\leftarrow\{0,1\}^{|s|}$.
Formally, for all PPT distinguishers $\mathcal{A}$:
\begin{equation}
  \left|\Prob_{s}[\mathcal{A}^{H(\cdot,s)}=1]
  -\Prob_{f}[\mathcal{A}^{f(\cdot)}=1]\right|
  \leq \mathrm{negl}(|s|).
  \label{eq:prf}
\end{equation}
\GS instantiates $H$ as HMAC-SHA256, which satisfies
Eq.~\eqref{eq:prf} under standard assumptions on the SHA-256
compression function~\cite{bellare1996keying}.
\end{definition}

\begin{definition}[Effective Channel Capacity]
\label{def:ceff}
Let $T:\calW\times\mathcal{K}\to\calW$ be a sanitization map and
let $\hat{P}=\mathcal{D}(T(W,k))$ denote the payload extracted from
the sanitized weights.
The \emph{effective channel capacity} of the covert channel under
defense $T$ with key $k$ is
\begin{equation}
  C_{\mathrm{eff}}(T,k)
  \;=\; \sup_{\mathcal{P}}\;\MI_{\mathcal{P}}\!\bigl(P\,;\,\hat{P}\bigr),
  \label{eq:ceff_formal}
\end{equation}
where the supremum is taken over all payload distributions
$\mathcal{P}$ on $\{0,1\}^B$, and $\MI_{\mathcal{P}}$ denotes
mutual information under that distribution.
Proposition~\ref{prop:erasure} establishes that
$C_{\mathrm{eff}}(\GS,s)=0$ for all $s\in\mathcal{K}$ unknown to the
attacker under the idealized PRF model.
\end{definition}

\begin{definition}[Payload Recovery]
\label{def:recovery}
Given an embedded payload $\mathbf{p}\in\{0,1\}^B$ and a decoded
sequence $\hat{\mathbf{p}}\in\{0,1\}^B$ obtained after applying
defense $T$, the \emph{payload recovery rate} is
\begin{equation}
  \mathrm{Rec}(\mathbf{p},\hat{\mathbf{p}})
  \;=\; \frac{1}{B}\sum_{i=1}^{B}\mathbf{1}[\hat{p}_i=p_i]
  \;=\; \BitAcc(\mathbf{p},\hat{\mathbf{p}}).
  \label{eq:rec}
\end{equation}
The \emph{Recovery Reduction} is then
$\RR = (\mathrm{Rec}_{\mathrm{pre}}-\mathrm{Rec}_{\mathrm{post}})
\times100\%$,
where $\mathrm{Rec}_{\mathrm{pre}}=\BitAcc(\mathbf{p},
\mathcal{D}(\mathcal{E}(W,\mathbf{p})))$ is the pre-defense recovery
rate (effectively $1.0$ for aligned injection at $x\leq16$) and
$\mathrm{Rec}_{\mathrm{post}}=\BitAcc(\mathbf{p},
\mathcal{D}(T(\mathcal{E}(W,\mathbf{p}),k)))$ is the post-defense
recovery rate.
\end{definition}

\begin{definition}[Stealth]
\label{def:stealth}
An embedding $\mathcal{E}(W,\mathbf{p})$ is \emph{$(\epsilon_\ell,
\epsilon_c,\epsilon_a)$-stealthy} if the following three conditions
hold simultaneously:
\begin{align}
  \Drel(W,\mathcal{E}(W,\mathbf{p})) &\leq \epsilon_\ell
    &&\text{(perturbation magnitude)},
    \label{eq:stealth_l2}\\
  1-\mathrm{Sim}_{\cos}(W,\mathcal{E}(W,\mathbf{p})) &\leq \epsilon_c
    &&\text{(directional preservation)},
    \label{eq:stealth_cos}\\
  |\mathrm{Acc}(f_W)-\mathrm{Acc}(f_{\mathcal{E}(W,\mathbf{p})})| &\leq \epsilon_a
    &&\text{(utility preservation)}.
    \label{eq:stealth_acc}
\end{align}
In the present study, the empirical capacity boundary of
$(10^{-4},2.6{\times}10^{-7},0.01)$-stealth is confirmed at $x{=}16$
across all four model presets, degrading to
$(8.5{\times}10^{-2},3.5{\times}10^{-3},0.01)$-stealth at $x{=}23$.
\end{definition}

\begin{definition}[Perturbation Field]
\label{def:pert_field}
Let $W^{(\ell)}\in\R^{N_\ell}$ be a target parameter tensor and let
$W^{(\ell)\prime}=T_\ell(W^{(\ell)})$ denote its sanitized version.
The \emph{perturbation field} of sanitization method $M$ on tensor
$\ell$ is the function
$\Delta^{(M,\ell)}:\{0,\ldots,N_\ell{-}1\}\to\R_{\geq0}$
defined by
\begin{equation}
  \Delta^{(M,\ell)}_i
  \;=\; \abs{w^{(\ell)\prime}_i - w^{(\ell)}_i}
  \;=\; g_i^{(\ell)}\cdot 2^{-23}\cdot 2^{e_i-127},
  \label{eq:pert_field}
\end{equation}
where $g_i^{(\ell)}$ is the overwrite value at position $i$ and the
indices follow the canonical row-major flattening order of the
parameter tensor.
The total perturbation field $\Delta^{(M)}$ of the full model is the
concatenation of $\Delta^{(M,\ell)}$ across all target tensors in
flattening order.
\end{definition}

\begin{definition}[Total-Variation Smoothness]
\label{def:tv_smooth}
The perturbation field $\Delta^{(M)}$ of method $M$ is
\emph{$\tau$-total-variation smooth} if
\begin{equation}
  E_{\mathrm{TV}}(M)
  \;=\; \sum_{i=0}^{N-2}\abs{\Delta^{(M)}_{i+1}-\Delta^{(M)}_i}
  \;\leq\; \tau.
  \label{eq:tv_def}
\end{equation}
A method with lower $E_{\mathrm{TV}}$ introduces more slowly varying
perturbations in the flattened weight order, reducing (i)~the
high-frequency spectral energy detectable by frequency-domain anomaly
scanners, and (ii)~the risk of exciting sharp curvature directions in
the task-loss Hessian.
The Gray-code property $\Ham(g_i^{(\ell)},g_{i+1}^{(\ell)})\leq1$
implies $\tau_{\GS}\leq N\cdot2^{-23}\cdot\E[2^{e_i}]$,
whereas \PMabbr with a period-2 pattern achieves
$\tau_{\PMabbr}=\Theta(N\Delta_{\max})$, giving
$\tau_{\GSabbr}\ll\tau_{\PMabbr}$.
\end{definition}

\begin{definition}[Memoryless Covert Channel]
\label{def:memoryless}
The covert channel
$(\mathcal{E},\mathcal{D},\calW,\{0,1\}^C)$ is
\emph{memoryless with respect to sanitization $T$}
if, for all $i\neq j$, the random variables
$\hat{P}_i=\mathcal{D}(T(W,k))_i$ and
$\hat{P}_j=\mathcal{D}(T(W,k))_j$
are conditionally independent given $P=(P_1,\ldots,P_B)$:
\begin{equation}
  \hat{P}_i \perp \hat{P}_j \mid P,\;k
  \qquad\forall\,i\neq j.
  \label{eq:memoryless}
\end{equation}
This condition holds for \GS because
$g_i^{(\ell)}=((i{+}\mathcal{O}_\ell)\oplus((i{+}\mathcal{O}_\ell)\gg1))\,\&\,m_x$
is a deterministic function of $i$ alone (given $\mathcal{O}_\ell$),
with no dependence on $g_j^{(\ell)}$ for $j\neq i$.
The post-defense extracted bits are therefore independent across
positions, justifying the memoryless BSC model.
\end{definition}

\begin{definition}[Chance-Level Recovery]
\label{def:chance}
For a binary payload channel with crossover probability $p$,
\emph{chance-level recovery} is the expected bit accuracy under
the maximum-entropy noise distribution:
\begin{equation}
  \mathrm{Rec}_{\mathrm{chance}}
  \;=\; \Prob(\hat{P}_i=P_i)
  \;=\; 1-p
  \;\xrightarrow{p\to1/2}\; \tfrac{1}{2}.
  \label{eq:chance}
\end{equation}
Proposition~\ref{prop:erasure} establishes that \GS drives $p\to1/2$
in the idealized PRF model, so $\mathrm{Rec}_{\mathrm{post}}\to1/2$
and $\RR\to50\%$.
``Near-chance sanitization'' in our empirical results refers to
$\mathrm{Rec}_{\mathrm{post}}\in[0.49,0.51]$, corresponding to a
channel crossover probability indistinguishable from $1/2$ at the
measured sample sizes.
\end{definition}

\section{Extended Trade-off Analysis}
\label{app:rq3_tradeoff}

Figure~\ref{fig:rq3-appendix-tradeoff} provides an extended view of the
security--utility trade-off across defenses under both baseline and
adaptive attacker settings. We place this figure in the appendix because the full configuration-level visualization is informative but visually
dense relative to the main paper figures.

The figure complements the main RQ3 analysis by showing the distribution
of recovery reduction (RR) and accuracy degradation across all evaluated
model, payload, and attacker configurations. It confirms that the high-level ordering observed in the main results remains stable across
attacker variants, while also exposing important differences in variance, operating-point consistency, and deployment robustness among defenses.

\begin{figure}[!ht]
  \centering
  \includegraphics[width=0.8\linewidth]{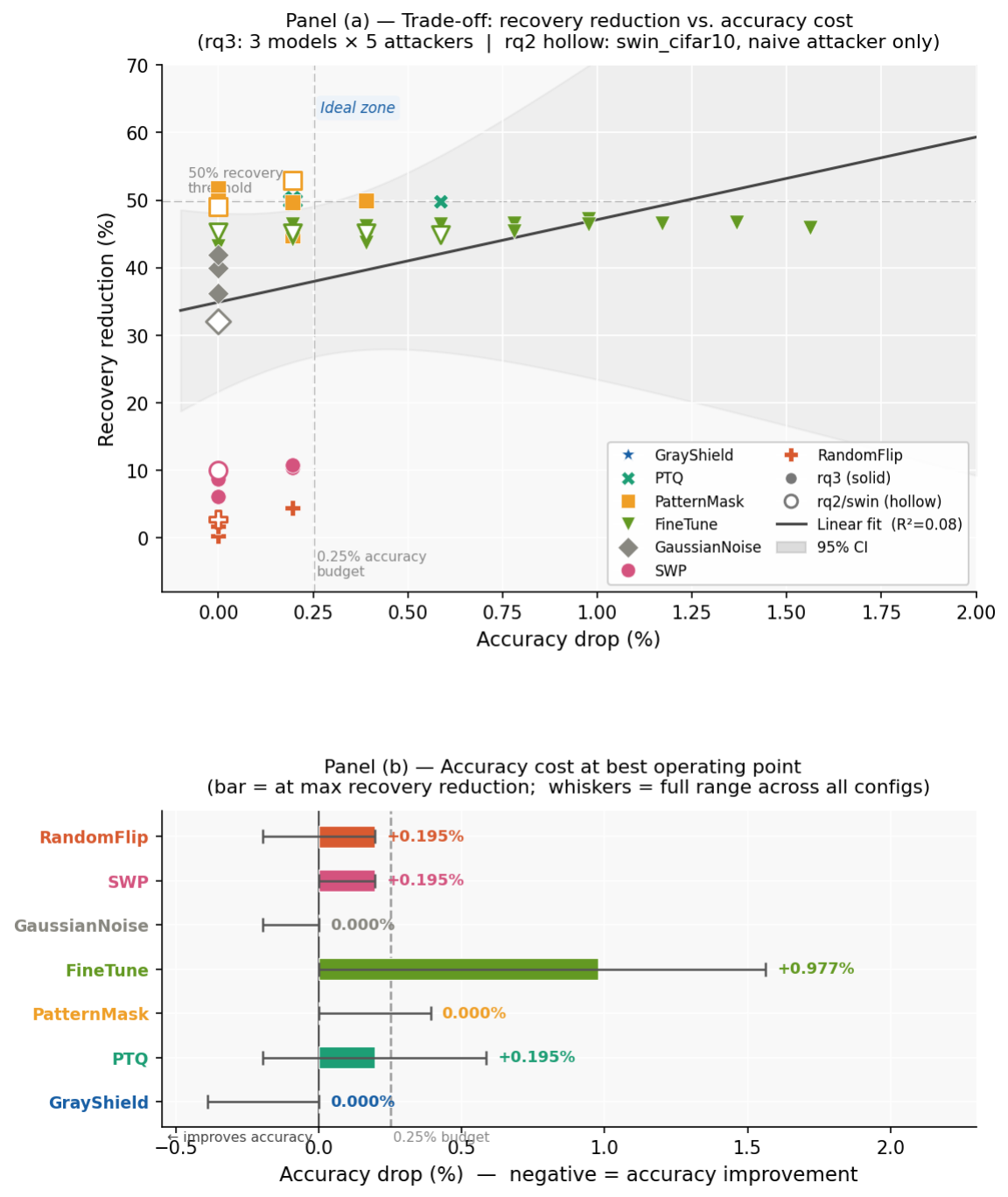}
  \caption{
\textbf{Security--utility trade-off across defenses.}
\textbf{Panel (a)} plots recovery reduction (RR) versus accuracy drop for all defense configurations. Solid markers denote configurations evaluated under the RQ3 adaptive attacker setting; hollow markers indicate the additional \texttt{swin\_cifar10} configuration from RQ2. Hollow markers show the additional \texttt{swin\_cifar10} configuration from RQ2. The horizontal dashed line marks the desired $50\%$ chance-level recovery threshold, while the vertical dashed line indicates the practical $0.25\%$ accuracy budget. \GSabbr{} consistently occupies the upper-left ``ideal zone,'' achieving near-chance RR with minimal accuracy degradation and low variance. \FTabbr{} achieves competitive RR only at substantially higher utility cost, whereas \GNabbr{}, \SWPabbr{}, and \RFabbr{} fail to provide effective sanitization. The weak global trend ($R^2=0.08$) indicates that stronger sanitization does not necessarily require larger accuracy degradation. 
\textbf{Panel (b)} shows the accuracy drop at each defense's best operating point, with whiskers indicating the full observed range across configurations. \GSabbr{} achieves near-optimal RR at effectively zero accuracy cost while maintaining substantially tighter variability than competing high-RR defenses.}
\label{fig:rq3-appendix-tradeoff}
\end{figure}

\section{Multi-Precision Threat Extension}
\label{app:multiprec}

The main evaluation targets FP32.
FP16~(1-5-10) and BF16~(1-8-7) reduce mantissa depth to
$m_{\mathrm{FP16}}=10$ and $m_{\mathrm{BF16}}=7$ bits respectively,
giving adjusted hiding capacity
$C_F=\sum_w\numel(w)\cdot\min(x,m_F)$.
For BERT-base encoder blocks at $x=5$ in FP16,
$C_{\mathrm{FP16}}\approx53$\,MB, sufficient for both evaluated
payloads.
Proposition~\ref{prop:erasure} extends to these formats unchanged,
since the Gray-code overwrite applies identically by clamping
$x\leq m_F$.
Each INT8 weight supports $x\leq2$ bits at quantization error
$\delta_i=q\cdot\mathrm{scale}_\ell$ with $q\leq1/128$, a negligible
stealth cost at $x=1$.
Empirical validation of reduced-precision formats is left as future
work.

\section{Bit Manipulation in FP32 Weights}
\label{app:bitops}

A float32 value stores sign (1 bit), exponent (8 bits), and mantissa
(23 bits).
All attack and defense operations act exclusively on mantissa LSBs,
preserving sign and exponent unchanged.

Let $u=\mathrm{bitcast}(w)\in\{0,1\}^{32}$ be the uint32 view of
weight $w$ and $m=(1\ll x)-1$ the LSB mask.
The three core operations are:
\begin{itemize}
  \item \textit{Extraction:} $u\;\&\;m$.
  \item \textit{Overwrite / sanitization:}
    $u'=(u\;\&\;\sim m)\;|\;b$,\quad $b\in[0,2^x-1]$.
  \item \textit{Random flip:}
    $u'=u\oplus r$ within the masked region, where $r$ is sampled
    from $\mathrm{Bernoulli}(p)$ independently per bit.
\end{itemize}
All three are fully vectorizable via
\texttt{Tensor.view(torch.int32)} on contiguous FP32 parameters,
completing in $O(N)$ time with $O(1)$ auxiliary state per tensor.

\section{Budget Definition}
\label{app:budget}

The \emph{bit-level budget} fixes the total number of modified bits
$B$, implying $K=\lfloor B/x\rfloor$ modified floats.
The \emph{float-level budget} modifies $K$ float elements, touching
$K\cdot x$ bits.
We use a bit-level budget throughout the paper for fairness across
different values of $x$.

\section{Additional Fidelity Metrics}
\label{app:extra_metrics}

\GS does not yield the single largest average Hamming distance;
that distinction is held by \PMabbr\ by a negligible margin.
However, \GS achieves essentially the same level of payload
corruption while preserving substantially better weight fidelity:
its average $W_1$ is about $25\times$ smaller than \PMabbr\ and
about $42\times$ smaller than \PTQabbr.
The main empirical advantage of \GS is therefore near-maximal
corruption with markedly cleaner distribution-level behaviour.

\section{Reproducibility}
\label{app:repro}

The anonymized replication package~\citep{anonymous2026grayshieldreplication}
contains the code, configuration grids, figure-generation scripts, curated
payload data, dataset metadata, structured result artifacts, and provenance
information needed to reproduce the main figures and tables in this paper.

The curated payload data include the two representative payloads used in the
completed main-paper runs---the low-entropy JavaScript sample
\texttt{c37{\ldots}a88.js.malware} and the high-entropy executable
\texttt{5704{\ldots}6d0.exe.malware}---together with four supplementary
payloads and a metadata manifest recording file type, size, entropy, and paper
role. The payload bundle is included in the anonymized replication package.

The repository also includes a structured results bundle under
\texttt{release\_results/}. This bundle separates the refreshed RQ1 rerun from
the completed RQ2--RQ4 results, and includes sanitization configuration files,
attacker encoder/decoder implementations, figure-generation utilities, and a
merged provenance manifest tracing each published result back to its source run.

\smallskip
\noindent
\textbf{Seeds:} run seed~1; evaluation seed~42;
$n_{\mathrm{eval}}=512$ per run.\\
\textbf{Software:} Python~3.13, PyTorch~2.9.1, Transformers~5.0.0.\\
\textbf{Hardware:} Ubuntu~22.04.5; $2\times$ Intel Xeon Gold~6430
(128 logical cores); 503\,GiB RAM;
$2\times$ NVIDIA RTX PRO~6000 Blackwell (96\,GB VRAM each).


\end{document}